\documentclass[10pt, conference, letterpaper]{IEEEtran}
\ifCLASSINFOpdf
 \usepackage[pdftex]{graphicx}
 \usepackage{subcaption}
\else
\fi
\usepackage{color}
\usepackage{amsmath}
\usepackage{amsthm}
\usepackage{amssymb}
\usepackage{bbm}
\usepackage{cite}
\usepackage{optidef}
\usepackage{steinmetz}
\usepackage[linesnumbered,lined,ruled,commentsnumbered]{algorithm2e}
\SetAlFnt{\small\sffamily}

\theoremstyle{definition}

\newtheoremstyle{mytheorem}
  {3pt}
  {3pt}
  {\itshape}
  {}
  {\itshape\bfseries}
  {.}
  {.5em}
  {\thmname{#1}\thmnumber{ #2} \thmnote{ {\the\thm@notefont(#3)}}}

\makeatother
\theoremstyle{mytheorem}

\newtheorem{theorem}{Theorem}

\newtheorem{lemma}[theorem]{Lemma}
\theoremstyle{remark}
\newtheorem*{remark}{Remark}

\makeatletter
\renewcommand*\env@matrix[1][*\c@MaxMatrixCols c]{%
  \hskip -\arraycolsep
  \let\@ifnextchar\new@ifnextchar
  \array{#1}}
\makeatother

\makeatletter
\patchcmd{\@makecaption}
  {\scshape}
  {}
  {}
  {}
\makeatother

\newcommand\Tstrut{\rule{0pt}{2.6ex}}         
\newcommand\Bstrut{\rule[-0.9ex]{0pt}{0pt}}   

\newcommand\blfootnote[1]{%
  \begingroup
  \renewcommand\thefootnote{}\footnote{#1}%
  \addtocounter{footnote}{-1}%
  \endgroup
}

\begin{document}
\title{Linear Block Coding for Efficient Beam Discovery
in Millimeter Wave Communication Networks}

\author{\IEEEauthorblockN{Yahia Shabara, C. Emre Koksal and Eylem Ekici}
\IEEEauthorblockA{Department of Electrical and Computer Engineering\\
The Ohio State University, Columbus, Ohio 43210\\
Email: \{shabara.1, koksal.2, ekici.2\}@osu.edu}
}
\maketitle



\IEEEpeerreviewmaketitle

\begin{abstract}
The surge in mobile broadband data demands is expected to surpass the available spectrum capacity below $6$ GHz. This expectation has prompted the exploration of millimeter wave (mm-wave) frequency bands as a candidate technology for next generation wireless networks. However, numerous challenges to deploying mm-wave communication systems, including channel estimation, need to be met before practical deployments are possible.
This work addresses the mm-wave channel estimation problem and treats it as a beam discovery problem in which locating beams with strong path reflectors is analogous to locating errors in linear block codes.
We show that a significantly small number of measurements (compared to the original dimensions of the channel matrix) is sufficient to reliably estimate the channel. We also show that this can be achieved using a simple and energy-efficient transceiver architecture.
\blfootnote{Eylem Ekici is supported in part by NSF grants CNS-1421576 and CNS-1731698.
C. Emre Koksal is supported in part by NSF grants CNS-1618566, CNS-1514260.}
\end{abstract}

\section{Introduction}
\label{Intro}
We investigate the problem of channel estimation in millimeter wave (mm-wave) wireless communication networks.
Mm-wave refers to the wavelength of electromagnetic signals at 30-300 GHz frequency bands.
At these high frequencies, channel measurement campaigns revealed that wireless communication channels exhibit very limited number of scattering clusters in the angular domain \cite{akdeniz2014millimeter, rangan2014millimeter, anderson2004building}. A \textit{cluster} refers to a propagation path or continuum of paths that span a small interval of transmit Angles of Departure (AoD) and receive Angles of Arrival (AoA).
Moreover, signal attenuation is very significant at mm-wave frequencies.
This motivates the use of large antenna arrays at the transmitter (TX) and receiver (RX) to provide high antenna gains that compensate for high path losses \cite{rappaport2013millimeter}.
Nevertheless, due to the high power consumption of mixed signal components, e.g., Analog to Digital Converters (ADCs) \cite{DSPforMmWave2016}, conventional digital transceiver architectures that employ a complete RF chain per antenna is not practical.
%
Hence, alternate architectures have been proposed for mm-wave radios with the objective of maintaining a close performance to channel capacity. Among the proposed solutions are the use of i) hybrid analog/digital beamforming \cite{alkhateeb2014channel,mendez2015channel,han2015large} and ii) fully digital beamforming with low resolution ADCs \cite{ImpactOfResOnPerf2015,highSNRcapacitySingleBitADC2014,AdaptiveOneBit_Rusu2015}.

For all proposed solutions, \textit{channel estimation} remains one of the most critical determinants of performance in communication.
Due to the large number of antennas at TX and RX,
estimation of the full channel gain matrix
may require a large number of measurements, proportional to the product of the number of transmit and receive antennas. This imposes a great burden on the estimation process. To address this issue, various methods have been used, the most prevalent among them, is compressed sensing \cite{Bajwa2010_compressedChannelSensing, ChEstSunAdaptiveCS2017, alkhateeb2014channel,ChEstScniter14,AdaptiveOneBit_Rusu2015}, which leverages channel sparsity.
Performance of compressed sensing based approaches is heavily dependent on the design of system (sensing) matrices. For instance, while random sensing matrices are known to perform well, in practice, sensing matrices involve the design of transmit and receive beamforming vectors and the choice of dictionary matrices\footnote{A dictionary matrix is used to express the channel in a sparse form.}.
Hence, purely random matrices have not been used in practice \cite{CS4Wireless_TipsAndTricks}. On the other hand, no design that involves deterministic sensing matrices has been considered for sparse channel estimation.

Despite the efforts, we do not have a full understanding of the dependence of channel estimation performance on the channel parameters and number of measurements.
In an effort to understand this relationship, the study in \cite{Alkhateeb_2015_HowManyMeasurements} proposed a multi-user mm-wave downlink framework based on compressed sensing in which the authors evaluate the achievable rate performance against the number of measurements.


In this work, we follow a different approach. We propose a systematic method in which we use sequences of error correction codes chosen in a way to control the channel estimation performance.
To demonstrate our approach, consider the following simple example.
Let a point to point communication channel be such that, there exists 3 possible receive AoA directions, only one of which may have a strong path to TX. We need to obtain the correct AoA at RX, if it exists.
Instead of exhaustively searching all 3 possible AoA directions, we alternatively measure signals from combined directions. For instance, by combining directions 1\&2 in one measurement and 2\&3 in the next measurement, we can find the AoA in just two measurements.
Specifically, four different scenarios might occur, namely, i) only the $1^{\text{st}}$, or ii) only the $2^{\text{nd}}$ measurement contains a strong path, iii) both $1^{\text{st}}$ and $2^{\text{nd}}$ measurements contain a strong path, and finally, iv) neither measurement reveals a strong path.
Interpretation of those cases is: AoA is in  i) direction 1, ii) direction 3, iii) direction 2, and iv) none exists.
Therefore, only 2 measurements are sufficient for beam detection instead of 3 that are needed for exhaustive search.


We will generalize this idea to develop a systematic method for beam detection, inspired by linear block coding.
Specifically, we show that linear block error correcting codes (LBC) possess favorable properties that fit in with the desirable behavior of sparse channel estimation. As a result, we are able to \textbf{i) provide rigorous criteria for solving the channel estimation problem}, \textbf{ii) significantly decrease the number of required measurements}, and \textbf{iii) utilize a fairly simple and energy-efficient transceiver architecture.}
We design the system using LBCs that leverage the fact that transmission errors are typically sparse in transmitted data streams, and hence, only a few number of erroneous bits need to be corrected per transmitted codeword.
Similarly, mm-wave channels are also sparse, i.e., only a small number of AoAs/AoDs carry strong signals.
LBCs can correct sparse transmission errors by identifying their location in a transmitted sequence (followed by flipping them).
We are inspired by LBC's ability to locate erroneous bits and exploit it to identify the AoAs/AoDs that carry strong signals (and their path gains) among all possible AoA/AoD values.
To this end, we exploit hard decision decoding of LBCs, in which the receiver obtains an \textit{error syndrome} that maps to one of the correctable error patterns.
An obtained error pattern determines the positions where errors have occurred.
Likewise, for channel estimation, the receiver will be designed to do a sequence of measurements that would result in a \textit{channel syndrome}. The resultant channel syndrome shall identify the positions (and values) of non-zero angular channel components.

Contributions of this work can be summarized as follows:
\begin{itemize}
\item	We set an analogy between beam discovery and channel coding to utilize low-complexity decoding techniques for efficient beam discovery.
\item	We provide rigorous criteria for setting the number of channel measurements based on the size of the channel and its sparsity level.
\item   We show that the number of measurements required for beam discovery is linked to the rate of a used linear block code. Hence, maximizing the rate of the underlying code is equivalent to minimizing the number of measurements.
\item	We develop a simple receiver architecture that enables us to measure signals arriving from multiple directions.
\end{itemize}

\textbf{Related Work:}
The main objective of mm-wave channel estimation is to find a mechanism that can reliably estimate the channel using as few measurements as possible.
For instance, in \cite{alkhateeb2014channel}, a compressed sensing based algorithm to estimate single-path channels is proposed and an upper bound on its estimation error is derived. Further, the authors propose a multipath channel estimation algorithm based on that of single-path channels. The proposed algorithms in \cite{alkhateeb2014channel} use an adaptive approach with a hierarchical codebook\footnote{A codebook refers to the set of all possible beamforming vectors.}
of increasing resolution. Similarly, the work in \cite{ChEstSunAdaptiveCS2017} proposes an adaptive compressive sensing channel estimation algorithm that accounts for off-the-grid AoAs and AoDs by using continuous basis pursuit \cite{CBP_2011} dictionaries.
Such adaptive algorithms divide the estimation process into stages and demand frequent feed back to the TX after each stage.
Hence, while the number of required measurements are shown to decrease, these methods may add a considerable overhead.

Other works like \cite{AgileMmWave_arxiv17,AgileMmWave_HotNets16} and \cite{RACE_2017} have proposed channel estimation algorithms using overlapped beam patterns.
For instance, the algorithm in \cite{RACE_2017} can estimate multipath channel components by sequentially estimating each path gain using an algorithm designed to estimate single-path channels followed by recursively removing the estimated paths' effect from subsequent measurements. Similar to \cite{alkhateeb2014channel,ChEstSunAdaptiveCS2017} adaptive beams with increasing resolution that require feedback to TX are used to refine the AoA/AoD estimates.

On the other hand, in \cite{AgileMmWave_arxiv17,AgileMmWave_HotNets16} beam alignment algorithms are proposed assuming a multipath mm-wave channel. The presented algorithms, with a high probability, can find the best beam in a logarithmic number of measurements (with respect to the total number of available AoA directions). Nonetheless, despite the possible existence of multiple paths, those algorithms are designed to find one path to TX.

Most research efforts in the field of mm-wave channel estimation use the magnitude and phase information of the acquired channel measurements. Nevertheless, if a carrier frequency offset (CFO) error occurs in the transceiver hardware, the phase information might be unreliable. Hence, the work in \cite{AgileMmWave_arxiv17,AgileMmWave_HotNets16,NonCoherentPathTrack_17} tackle this problem by ignoring the obtained phase information.
Similar to \cite{AgileMmWave_arxiv17,AgileMmWave_HotNets16}, the solution in \cite{NonCoherentPathTrack_17} can only obtain one (dominant) path between TX and RX, but it uses a compressed sensing based technique.
The CFO problem is tackled in \cite{myers2017compressive} by considering it as a variable to be estimated. Specifically, a third-order tensor is formulated in which two dimensions account for AoA/AoD and the third one is for CFO. The formulated tensor is shown to be sparse and is estimated using compressed sensing techniques.

While the power consumption problem is commonly alleviated using analog or hybrid beamforming transceivers, the work in \cite{ImpactOfResOnPerf2015,highSNRcapacitySingleBitADC2014,AdaptiveOneBit_Rusu2015,alkhateeb2014mimo}
consider using low resolution (single-bit) ADCs in a fully digital architecture. The work in \cite{ChEstScniter14,barati2015directional} study the channel estimation problem using such architectures. 


\subsubsection*{Notations} A vector and a matrix are denoted by $\boldsymbol{x}$ and $\boldsymbol{X}$, respectively, while $x$ denotes a scalar or a complex number depending on the context. The transpose, conjugate transpose and frobenius norm of $\boldsymbol{X}$ are given by $\boldsymbol{X}^T$, $\boldsymbol{X}^H$ and $\left\Vert \boldsymbol{X} \right\Vert_F$, respectively.
The sets of real and complex numbers are $\mathbb{R}$ and $\mathbb{C}$.
The $k {\times} k$ identity matrix is $\boldsymbol{I_k}$.
A set is denoted by $\mathcal{X}$, while $|\mathcal{X}|$ is its cardinality. Finally, $\mathbbm{1}()$ is the indicator function.

\section{Motivating Example}
\label{illustrative}
To elaborate, we present the following example:
consider a point to point communication link between a TX with single antenna ($n_t = 1$) and RX with $n_r = 15$ antennas. Therefore, the vector of channel gains\footnote{Let all the channels have one single significant tap.},
$\boldsymbol{q}$, is a $15 {\times} 1$ vector,
and its corresponding angular (virtual) channel, $\boldsymbol{q}^a$, is a vector of the same size and can be derived using the DFT matrix $\boldsymbol{U_r}$ as $\boldsymbol{q}^a =\boldsymbol{U}^H_{\boldsymbol{r}} \: \boldsymbol{q}$ \cite{sayeed2002deconstructing} (this is merely a linear transformation that maps the sequence of channel gains into a sequence of gains from different AoAs. This mapping will be presented in more detail in Section \ref{SystemModel}).
Assume a single-path channel, i.e., the channel has only one cluster with a single path in it.
Let the path gain be denoted by $\alpha$. For simplicity assume $\alpha = 1$.
Further, let us assume perfect sparsity such that the AoA is along one of the directions defined in the DFT matrix $\boldsymbol{U_r}$, i.e., the channel path will only contribute to one angular bin.
Finally, let us also neglect the channel noise.

\begin{figure}[t]
\centering
\includegraphics[width=0.5\linewidth]{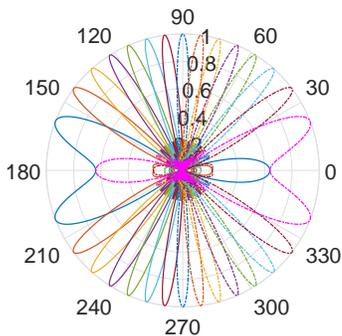}
\caption{Beam patterns of all possible angular directions}
\label{fig:AllBeams} 
\end{figure}

Based on the channel description above, we get an angular channel vector of the form
\begin{equation}
\boldsymbol{q}^a = 
 \begin{pmatrix}
  q^a_0 & q^a_1 &  \hdots &  q^a_{14}
 \end{pmatrix}^T,
\end{equation}
such that $q^a_i \in \{0,1\}$ and the number of non-zero elements in $\boldsymbol{q}^a$ is $1$. Any component of $\boldsymbol{q}^a$ can be measured using one of the beam patterns shown in Fig. \ref{fig:AllBeams}.

\textbf{Objective:}
Suppose the transmitter sends pilot symbols of the form $x {=} 1$. Thus, the received vector $\boldsymbol{y}$ of size $15 {\times} 1$ can be obtained as
\begin{equation}
\boldsymbol{y}   =  \boldsymbol{q} x   =  \boldsymbol{q} \Longleftrightarrow \boldsymbol{y}^a    =  \boldsymbol{q}^a
\end{equation}
where $\boldsymbol{y}^a$ is the received vector in the angular domain. So, with change of basis, we can think of $\boldsymbol{q}^a$ as a received sequence with just one non-zero component. To identify the position of this non-zero component, the receiver performs a sequence of channel measurements. Let $y_{s_i}$ denote the $i^{th}$ measurement such that
\begin{equation}
y_{s_i} = \boldsymbol{w}_i^H \boldsymbol{y} = \boldsymbol{w}_i^H \boldsymbol{q},
\end{equation}
where $\boldsymbol{w}_i$ denotes the $i^{th}$ receive (rx-)combining vector.

Our aim is to design channel measurements (i.e., $\boldsymbol{w_i}$'s) such that the correct AoA is identified using the minimum number of measurements.

\textbf{Proposed Solution:}
We consider this non-zero component to be an anomaly to a normally all-zero $15$-bin angular channel.
Hence, the goal of identifying its position is analogous to finding the most likely $1$-bit error pattern of a $15$-bit codeword in a linear block code.
Now, we need to identify an error correction code with codewords of length $15$ and with $1$-bit error correction capability \cite{van2012introduction}. Hence, we can use the binary $(15,11,3)$ Hamming code with parity check matrix $\boldsymbol{H}$ of size $4 {\times} 15$ and given by
\begin{equation}
  \boldsymbol{H}{=}
  \begin{pmatrix}
  1 & 0 & 0 & 0 & 1 & 0 & 0 & 1 & 1 & 0 & 1 & 0 & 1 & 1 & 1 \\
  0 & 1 & 0 & 0 & 1 & 1 & 0 & 1 & 0 & 1 & 1 & 1 & 1 & 0 & 0 \\
  0 & 0 & 1 & 0 & 0 & 1 & 1 & 0 & 1 & 0 & 1 & 1 & 1 & 1 & 0 \\
  0 & 0 & 0 & 1 & 0 & 0 & 1 & 1 & 0 & 1 & 0 & 1 & 1 & 1 & 1 \\
  \end{pmatrix}
\end{equation}
where $h_{i,j}$ represents the component at the intersection of row $i$ and column $j$ of $\boldsymbol{H}$. Using hard decision decoding of LBCs, error syndrome vectors of length $4$ are obtained. Every possible syndrome vector maps to only one correctable error pattern\footnote{A correctable error pattern of a ($15,11,3$) Hamming code is any $15{\times}1$ binary vector that contains only one '1' (at the error's position).}.
Similarly, for channel estimation, several measurements should be performed at RX where each measurement mimics the behavior of a corresponding element in the error syndrome vector. Each measurement boils down to adding signals from a subset of the available $15$ directions. Since each measurement can either include the direction of the incoming strong path of gain $\alpha=1$ or no strong paths at all, then the elements of the channel syndrome vector are in $\{0,1\}$.

\begin{figure*}[t]
\centering
\includegraphics[width=1\linewidth]{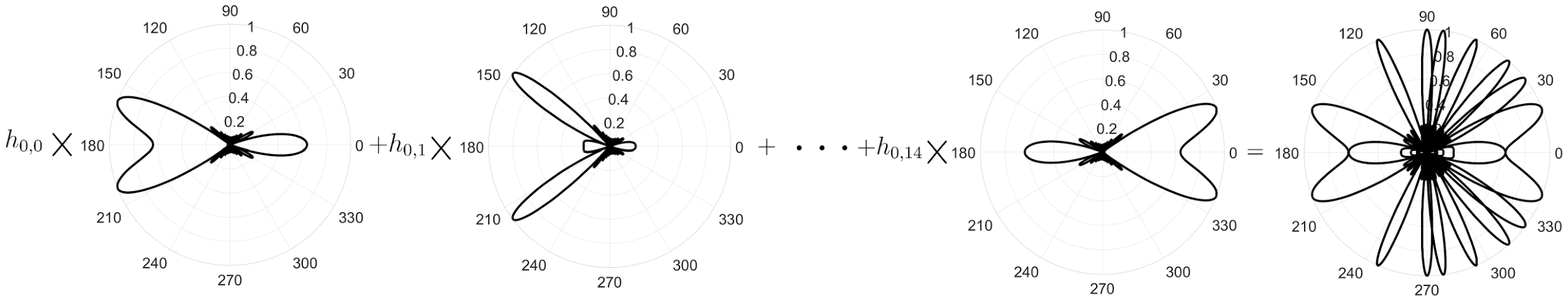}
\caption{Beam pattern of receive combining vector $\boldsymbol{w_0}$}
\label{fig:RxComb} 
\end{figure*}

For every measurement $y_{s_i}$, we design $\boldsymbol{w}_i$ based on the entries of the $i^{th}$ row of $\boldsymbol{H}$ such that: if $h_{i,j} = 1$, then we include the beam pattern that points to direction $j$ in $\boldsymbol{w}_i$.
For example, the $0^{th}$ row of $\boldsymbol{H}$ is given by $[1 \: 0 \: 0 \: 0 \: 1 \: 0 \: 0 \: 1 \: 1 \: 0 \: 1 \: 0 \: 1 \: 1 \: 1]$. Hence, $\boldsymbol{w}_0$ should include beam patterns pointing to the set of directions $\{0,4,7,8,10,12,13,14\}$.

Fig. \ref{fig:RxComb} illustrates this operation for $\boldsymbol{w}_0$. We can see that the resultant beam pattern of $\boldsymbol{w}_i$ combines signals coming from a set of selected directions dictated by the $i^{th}$ row of $\boldsymbol{H}$.
We call the obtained measurement vector, $\boldsymbol{y_s}$, the \textit{channel syndrome} which is analogous to error syndromes in hard decision decoding of LBCs.
Then, a table that maps every possible channel syndrome to a unique corresponding channel can be constructed.
Table \ref{table:SyndromeMap} shows this mapping.

\begin{table}[t]
\caption{\small Mapping of channel syndromes to angular channels}
\label{table:SyndromeMap}
\centering
 \begin{tabular}{||c | c||} 
 \hline
 Channel Syndrome $\boldsymbol{y}^T_{\boldsymbol{s}}$ & Angular Channel ${\boldsymbol{q}^a}^T$ \Tstrut \Bstrut \\ 
 \hline\hline
$[0 \: 0 \: 0 \: 0]$ & $[0 \: 0 \: 0 \: 0 \: 0 \: 0 \: 0 \: 0 \: 0 \: 0 \: 0 \: 0 \: 0 \: 0 \: 0]$\\ 
 \hline
$[1 \: 0 \: 0 \: 0]$ & $[1 \: 0 \: 0 \: 0 \: 0 \: 0 \: 0 \: 0 \: 0 \: 0 \: 0 \: 0 \: 0 \: 0 \: 0]$\\
 \hline
$[0 \: 1 \: 0 \: 0]$ & $[0 \: 1 \: 0 \: 0 \: 0 \: 0 \: 0 \: 0 \: 0 \: 0 \: 0 \: 0 \: 0 \: 0 \: 0]$\\
 \hline
$[0 \: 0 \: 1 \: 0]$ & $[0 \: 0 \: 1 \: 0 \: 0 \: 0 \: 0 \: 0 \: 0 \: 0 \: 0 \: 0 \: 0 \: 0 \: 0]$\\
 \hline
$[0 \: 0 \: 0 \: 1]$ & $[0 \: 0 \: 0 \: 1 \: 0 \: 0 \: 0 \: 0 \: 0 \: 0 \: 0 \: 0 \: 0 \: 0 \: 0]$\\
 \hline
$[1 \: 1 \: 0 \: 0]$ & $[0 \: 0 \: 0 \: 0 \: 1 \: 0 \: 0 \: 0 \: 0 \: 0 \: 0 \: 0 \: 0 \: 0 \: 0]$\\
 \hline
$[0 \: 1 \: 1 \: 0]$ & $[0 \: 0 \: 0 \: 0 \: 0 \: 1 \: 0 \: 0 \: 0 \: 0 \: 0 \: 0 \: 0 \: 0 \: 0]$\\
 \hline
$[0 \: 0 \: 1 \: 1]$ & $[0 \: 0 \: 0 \: 0 \: 0 \: 0 \: 1 \: 0 \: 0 \: 0 \: 0 \: 0 \: 0 \: 0 \: 0]$\\
 \hline
$[1 \: 1 \: 0 \: 1]$ & $[0 \: 0 \: 0 \: 0 \: 0 \: 0 \: 0 \: 1 \: 0 \: 0 \: 0 \: 0 \: 0 \: 0 \: 0]$\\
 \hline
$[1 \: 0 \: 1 \: 0]$ & $[0 \: 0 \: 0 \: 0 \: 0 \: 0 \: 0 \: 0 \: 1 \: 0 \: 0 \: 0 \: 0 \: 0 \: 0]$\\
 \hline
$[0 \: 1 \: 0 \: 1]$ & $[0 \: 0 \: 0 \: 0 \: 0 \: 0 \: 0 \: 0 \: 0 \: 1 \: 0 \: 0 \: 0 \: 0 \: 0]$\\
 \hline
$[1 \: 1 \: 1 \: 0]$ & $[0 \: 0 \: 0 \: 0 \: 0 \: 0 \: 0 \: 0 \: 0 \: 0 \: 1 \: 0 \: 0 \: 0 \: 0]$\\
 \hline
$[0 \: 1 \: 1 \: 1]$ & $[0 \: 0 \: 0 \: 0 \: 0 \: 0 \: 0 \: 0 \: 0 \: 0 \: 0 \: 1 \: 0 \: 0 \: 0]$\\
 \hline
$[1 \: 1 \: 1 \: 1]$ & $[0 \: 0 \: 0 \: 0 \: 0 \: 0 \: 0 \: 0 \: 0 \: 0 \: 0 \: 0 \: 1 \: 0 \: 0]$\\
 \hline
$[1 \: 0 \: 1 \: 1]$ & $[0 \: 0 \: 0 \: 0 \: 0 \: 0 \: 0 \: 0 \: 0 \: 0 \: 0 \: 0 \: 0 \: 1 \: 0]$\\
 \hline
$[1 \: 0 \: 0 \: 1]$ & $[0 \: 0 \: 0 \: 0 \: 0 \: 0 \: 0 \: 0 \: 0 \: 0 \: 0 \: 0 \: 0 \: 0 \: 1]$\\ 
\hline
\end{tabular}
\end{table}

In this example, we are able to estimate the channel based on only $4$ measurements as opposed to $15$, which is the number of measurements with exhaustive search.
Important aspects of our proposed method include the choice of codes, the design of precoding and rx-combining measurement vectors, the effect of variable gains and phases of different paths and the occurrence of measurement errors.

\begin{remark}[Receiver Architecture] Note that, to achieve beam patterns similar to the one shown in Fig. \ref{fig:RxComb}, the receiver architecture needs to be a bit different from those of classical analog/hybrid beamforming architectures. Specifically, in addition to low-noise amplifiers (LNA) typically placed at the output of each antenna, we will need to add controllable low-power amplifiers, as well. The resultant architecture is still quite simple (see Fig. \ref{fig:architecture}). That is, besides the low-power amplifiers, the proposed architecture is similar to those of simple analog beamforming. Moreover, low-resolution ADCs can be used which mitigates the high power consumption problem associated with high-resolution ADCs.
\end{remark}

\textbf{Motivation for LBC-inspired approach:}
LBCs are designed to discover and correct a certain maximum number of errors in a codeword of a specified length. This objective is achieved by adding redundant \textit{parity check} bits to the original information sequence.
What makes our devised approach attractive is that the number of measurements needed for channel estimation can be shown to be equal to the number of parity bits of some corresponding code. Hence, we can control the estimation performance via appropriate code selection.
In this work, we will propose a method to specify the number of necessary channel measurements as a function of the rate of the underlying code.


\section{System Model}
\label{SystemModel}

\begin{figure}[t]
\centering
\includegraphics[width=1\linewidth]{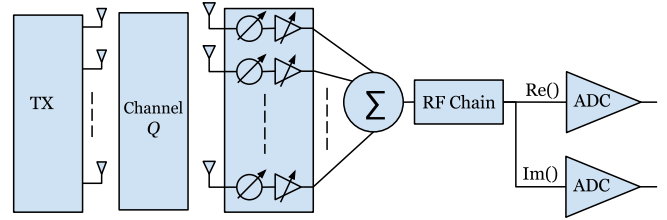}
\caption{Hardware Block Diagram: Every antenna is connected to a phase shifter and low-power variable gain amplifier. Then, all outputs are combined using an adder and passed to an RF chain with in-phase and quadrature channels.}
\label{fig:architecture} 
\end{figure}

Consider a point-to-point millimeter-wave wireless communication system with a transmitter (TX) equipped with $n_t$ antennas and a receiver (RX) with $n_r$ antennas placed at fixed locations. Uniform Linear Arrays (ULA) are assumed at both TX and RX where each antenna element is connected to a phase shifter and a variable gain amplifier. A single RF chain at the receiver, with in-phase (I) and quadrature (Q) channels, is fed through a linear combiner (see Fig. \ref{fig:architecture}). Only two mid-tread ADCs, with $2^b{+}1$ quantization levels, are utilized, i.e., quantization levels are $\{-2^{b-1}, \dots,-1,0,1,\dots,2^{b-1}\}$.

We adopt a single-tap channel model where $\boldsymbol{Q} \in \mathbb{C}^{ n_r { \times } n_t}$ denotes the channel matrix between TX and RX.
Assume that the channel has $L$ clusters, where each cluster contains a single path with attenuation $\alpha_l$, AoD $\theta_l$, and AoA $\phi_l$.
The channel is assumed to be sparse such that $L \ll n_t, n_r$.
Let $\alpha_l^b \in \mathbb{C}$ denote the baseband channel gain and is defined as
\begin{equation}
\alpha_l^b {=} \alpha_l \: \sqrt{n_t n_r} \: e^{- j\frac{2 \pi \rho_l}{\lambda_c}}
\end{equation}
where $\rho_l$ is the length of path $l$ and $\lambda_c$ is the carrier frequency.
The angular cosines of AoD and AoA associated with path $l$ are denoted by $\Omega_{tl}$ and $\Omega_{rl}$, respectively. The transmit and receive spatial signatures along the direction $\Omega$ is given by $\boldsymbol{e_t}(\Omega)$ and $\boldsymbol{e_r}(\Omega)$ such that
\begin{equation}
\boldsymbol{e_t}(\Omega) {=} \frac{1}{\sqrt{n_t}}
  \begin{pmatrix}
  1    							    	  \\
  e^{-j2\pi   \Delta_t \Omega}  	  \\
  e^{-j2\pi 2 \Delta_t \Omega}  	  \\
  \vdots								  \\
  e^{-j2\pi (n_t-1) \Delta_t \Omega}  \\
  \end{pmatrix},
\end{equation}
where  $\boldsymbol{e_r}(\Omega)$ has a similar definition to $\boldsymbol{e_t}(\Omega)$, and
$\Delta_t$ and $\Delta_r$ are the antenna separations at TX and RX normalized by the wavelength $\lambda_c$. Thus, 
$\boldsymbol{Q}$ can be written as
\begin{equation}
\boldsymbol{Q} = \sum_{l=1}^{L} \alpha_l^b \: \boldsymbol{e_r}(\Omega_{rl}) \: \boldsymbol{e}^H_{\boldsymbol{t}}(\Omega_{tl}).
\end{equation}
We define $\boldsymbol{U_t}$ and $\boldsymbol{U_r}$ as the unitary Discrete Fourier Transform (DFT) matrices whose columns constitute an orthonormal basis for the transmit and receive signal spaces $\mathbb{C}^{n_t}$ and $\mathbb{C}^{n_r}$, respectively. $\boldsymbol{U_t}$ (and similarly $\boldsymbol{U_r}$) is given by
\begin{equation}
\boldsymbol{U_t} = 
  \begin{pmatrix}
  \boldsymbol{e_t}(0) & \boldsymbol{e_t}(\frac{1}{L_t}) & \dots & \boldsymbol{e_t}(\frac{n_t-1}{L_t})
  \end{pmatrix},
\end{equation}
where $L_t$ and $L_r$ are the normalized lengths of the transmit and receive antenna arrays such that $L_t = n_t \Delta_t$ and $L_r = n_r \Delta_r$.
Let $\boldsymbol{Q^a}$ be the channel matrix in the angular domain \cite{sayeed2002deconstructing}, where
\begin{equation} \label{AngularChannel}
\boldsymbol{Q^a} = \boldsymbol{U}^H_{\boldsymbol{r}} \boldsymbol{Q} \boldsymbol{U_t}.
\end{equation}
The rows and columns of $\boldsymbol{Q}^a$ divide the channel into resolvable RX and TX bins, respectively.
Further, we assume a perfect sparsity model in which AoDs $\theta_l$, and AoA $\phi_l$, are along the directions defined in $\boldsymbol{U_t}$ and $\boldsymbol{U_r}$ \cite{ChEstScniter14,alkhateeb2014channel,RACE_2017}. Hence, each channel cluster will only contribute to a single pair of TX and RX bins. Therefore, $\boldsymbol{Q}^a$ has a maximum of L non-zero TX and RX bins.

The baseband channel model is given by
\begin{equation}
\boldsymbol{y} = \boldsymbol{Q} \boldsymbol{x} + \boldsymbol{n},
\end{equation}
where $\boldsymbol{x} \in \mathbb{C}^{n_t}$ is the transmitted signal, $\boldsymbol{y} \in \mathbb{C}^{n_r}$ is the received signal and $\boldsymbol{n} \sim \mathcal{CN}(\boldsymbol{0} , N_0 \boldsymbol{I}_{n_r} )$
is an i.i.d. complex Gaussian noise vector.

Let $\boldsymbol{f} \in \mathbb{C}^{n_t} $ and $\boldsymbol{w} \in \mathbb{C}^{n_r}$ be the precoding and rx-combining vectors, respectively. The transmit signal $\boldsymbol{x}$ is given by $\boldsymbol{x} = \boldsymbol{f} s$ where $s$ is the transmitted symbol with average power $\mathbb{E}(ss^H) {=} P$. After the receiver applies the rx-combining vector $\boldsymbol{w}$, the resultant symbol $y_s$ can be written as
\begin{equation}
y_s = \boldsymbol{w}^H \boldsymbol{Q} \boldsymbol{f} s + \boldsymbol{w}^H \boldsymbol{n}
\end{equation}
where $y_s$ constitutes a single unit measurement obtained using specific $\boldsymbol{f}$ and $\boldsymbol{w}$ vectors.
We assume that $\boldsymbol{Q}$ remains fixed throughput the entire estimation process. The noise component $\boldsymbol{w}^H \boldsymbol{n}$ normalized by $\left\Vert \boldsymbol{w} \right\Vert$ is also a complex gaussian random variable such that $\frac{\boldsymbol{w}^H}{\left\Vert \boldsymbol{w} \right\Vert} \boldsymbol{n} \sim \mathcal{CN}(0,N_0)$.
We define the signal to noise ratio (\textit{SNR}) on a per path basis such that \textit{SNR} of path $l$ is given by
\begin{equation}
\label{eqn:SNRl}
\textit{SNR}_l = \frac{P}{N_0} \lvert \alpha_l^b \rvert^2.
\end{equation}
Note that the actual received \textit{SNR} depends on all path gains included in a measurement.
Finally, we define a \textit{detectable path} as a path $l$ for which $\lvert \alpha_l^b \rvert$ is such that the received signal power on that path is above a minimum detectable signal power (MDS).

\section{Problem Statement}
\label{ProblemStatement}

Suppose a maximum number of $L$ clusters need to be discovered in the channel where $L \ll n_t,n_r$. Under the prefect sparsity assumption, $\boldsymbol{Q}^a$ has a maximum of $L$ non-zero RX and TX angular bins.
Our objective is to identify the angular positions at which channel clusters exist and identify their path gain values using the \textbf{least possible number of measurements}. Let the number of measurements be $m$ such that each measurement, $y_{s_{i,j}}$, is obtained using the precoder $\boldsymbol{f_j}$ and rx-combiner $\boldsymbol{w_i}$. Measurements take the form $y_{s_{i,j}} = \boldsymbol{w_i}^H \boldsymbol{Q} \boldsymbol{f_j} s + \boldsymbol{w_i}^H \boldsymbol{n}$. 
Let $\xi()$ be a mapping function that takes in the measurements $\{y_{s_{i,j}}\}_{\forall i,j}$ as inputs and returns the estimated channel $\widehat{\boldsymbol{Q}}^{\boldsymbol{a}}$.
We stack the measurements $\{y_{s_{i,j}}\}_{\forall j}$ in a single syndrome vector such that $\boldsymbol{y_{s_j}} = [y_{s_{0,j}} \: y_{s_{1,j}} \dots y_{s_{{m_1-1},j}}]^T$.
Our design variables are the precoding vectors $\boldsymbol{f_j}$, rx-combining vectors $\boldsymbol{w_i}$, the number of measurements $m$, the mapping function $\xi()$, and the transmitted symbol power $P$.

In its essence, solving this problem boils down to finding the optimal set of measurements $\{y_{s_{i,j}}\}_{\forall i,j}$ and the mapping function $\xi()$ such that $\boldsymbol{Q}^a$ can be estimated using the minimum number of measurements.
For ease of explanation, we first consider a channel with a single transmit antenna and $n_r$ receive antennas. Therefore, no precoding is needed and the design of measurements is reduced to designing the rx-combining vectors $\boldsymbol{w_i}$.
Recall that in the motivating example in Section \ref{illustrative}, we dealt with a special case of $n_r {\times} 1$ channels where we sought to find the direction of arrival of a channel with a single path of known gain, $\alpha = 1$. In the general case, we should consider arbitrary path gains $\alpha \in \mathbb{C}$ and channels with multiple paths.

\section{Beam Discovery}
\label{BeamDiscovery}
In this section, we present our proposed solution. As an initial step, we solve a simplified version of the problem where communication channels have a single transmit antenna and multiple receive antennas. Afterwards, we will build on it to provide the solution for general channels with multiple transmit and receive antennas.

\subsection{Beam Detection using LBC-inspired approach}
To identify the exact number of measurements and their corresponding design, we follow a decoding-like approach of
LBC\footnote{\label{LBC}In channel coding, the convention is to use row vectors. Thus, let $\boldsymbol{x}$ and $\boldsymbol{c}$ be $1{\times}k$ and $1{\times}n$ binary row vectors that represent an information sequence and its corresponding codeword of an LBC, respectively. Also let $\boldsymbol{r} {=} \boldsymbol{c} {+} \boldsymbol{e}$ be a received sequence corrupted by $1{\times}n$ error pattern $\boldsymbol{e}$. To decode $\boldsymbol{r}$, we calculate an error syndrome  vector $\boldsymbol{s}$, of size $1{\times}n{-}k$, such that $\boldsymbol{s} {=} \boldsymbol{r} \boldsymbol{H}^T$, where $\boldsymbol{H}$ is the parity check matrix of the used LBC. Then, a most likely error pattern $\hat{\boldsymbol{e}}$ can be uniquely identified by $\boldsymbol{s}$ using a look-up table called the \textit{standard array}. Finally, the decoded codeword is obtained using $\hat{\boldsymbol{c}}{=}\boldsymbol{r}{-}\hat{\boldsymbol{e}}$. A decoding error occurs if the number of errors, identified using $1$'s in $\boldsymbol{e}$, is beyond the error correction capability of the used code, denoted by $e_n$.
Note that in this context, all vectors, matrices and math operations are over GF(2).}.
First, we need to find an LBC, $C$, that has an error correction capability $e_n$ such that i) the maximum number of clusters in the channel, $L$, is equal to $e_n$ and ii) the length of its codewords $n$ is equal to the number of antennas $n_r$ ($n_r $ is also the number of resolvable directions).
The code $C$ has a parity check matrix $\boldsymbol{H}$ which represents the link between channel decoding and beam detection problems.
Binary codes deal with data and error sequences defined over the finite field $GF(2)$, i.e., addition and multiplication operations are defined over $GF(2)$ with binary inputs and outputs, i.e., $1$'s and $0$'s.
However, mm-wave channel parameters are defined over the complex numbers field $\mathbb{C}$. Therefore, to account for arbitrary path gains, we should be able to extend this concept to $\mathbb{C}$.

\textbf{Although $\boldsymbol{H}$ is defined over $GF(2)$, we interpret its $'1'$ and $'0'$ entries as real numbers.}
Then, similar to channel decoding, we seek to obtain a channel syndrome, $\boldsymbol{y_s}$, such that
${(\boldsymbol{y}_s)}^T {=} {(\boldsymbol{q}^a)}^T \boldsymbol{H}^T \Longrightarrow \boldsymbol{y}_s {=} \boldsymbol{H} \boldsymbol{q}^a$.
This matrix multiplication can be realized using channel measurements such that each measurement gives one component in $\boldsymbol{y_s}$.
Measurements $\{ y_{s_i}\}_{\forall i}$ make up the components of the channel syndrome vector $\boldsymbol{y_{s}}$. Then, we need to find a mapping function $\xi()$ that takes in the channel syndrome vector $\{\boldsymbol{y_{s}}\}$ as an input and returns the estimated channel $\widehat{\boldsymbol{q}}^{\boldsymbol{a}}$.
The position of each non-zero component in $\widehat{\boldsymbol{q}}^{\boldsymbol{a}}$ identifies a path's AoA, and its value identifies its baseband path gain.
Finally, for this to work, we need to show that such channel measurements provide one-to-one mapping to the channel. In other words, $\boldsymbol{y_s}$ must be a \textit{sufficient statistic} for estimating the channel.
In Section \ref{SufficientStatistic}, we will show that our design results in the sufficient statistic we seek to achieve.

\begin{remark}[Number of Measurements]
The solution we obtain is dependent on channel parameters, namely, the number of antennas and the sparsity level of the channel.
That is, at a fixed sparsity level, i.e., fixed number of clusters $L$, a larger number of antennas necessitates more channel measurements. In other other words, the high resolution realized by large $n_r$ comes at a price of an increased number of measurements.
Similarly, at fixed $n_r$, more channel clusters involve more measurements for correct channel estimation.
\end{remark}

\subsection{Measurements Design}
\label{MeasurementsDesign}
Recall that each component in $\boldsymbol{q}^a$ represents a resolvable angular direction at the receiver. Let each resolvable direction be given an identification number ($\textit{dir}_{\textit{rx}}$\textit{\#i}).
Also let $\textit{beam}_{\textit{rx}}$\textit{\#i} denote the beam pattern pointing to $\textit{dir}_{\textit{rx}}$\textit{\#i}, i.e., a signal coming from $\textit{dir}_{\textit{rx}}$\textit{\#i} can be individually measured using $\textit{beam}_{\textit{rx}}$\textit{\#i} (similar to beam patterns in Fig. \ref{fig:AllBeams}).

Now, we seek to obtain $\boldsymbol{y_s} {=} \boldsymbol{H} \boldsymbol{q^a}$ using careful design of $\boldsymbol{w_i}$'s (note that we drop the noise terms for clarity), i.e.,
\begin{equation}
\boldsymbol{y_s} = 
 \begin{pmatrix}
  y_{s_0} \\
  y_{s_1} \\
  \vdots\\
  y_{s_{m{-}1}}
 \end{pmatrix}
 =
  \begin{pmatrix}
  \boldsymbol{w}^H_{\boldsymbol{0}} \boldsymbol{q} \\
  \boldsymbol{w}^H_{\boldsymbol{1}} \boldsymbol{q} \\
  \vdots\\
  \boldsymbol{w}^H_{\boldsymbol{m{-}1}} \boldsymbol{q} \\
  \end{pmatrix}
  \equiv
  \boldsymbol{H} \boldsymbol{q^a}.
\end{equation}
To achieve this, each rx-combining vector $\boldsymbol{w_i}$ is designed as a multi-armed beam, i.e., composed of several sub-beams similar to the beam pattern in Fig. \ref{fig:RxComb}. The sub-beams included in each $\boldsymbol{w_i}$ are identified by the $i^{th}$ row of the matrix $\boldsymbol{H}$.
That is, only if $h_{i,j}$, the intersection of the $i^{th}$ row and $j^{th}$ column, is $=1$, do we include $\textit{beam}_{\textit{rx}}$\textit{\#j} as a sub-beam in $\boldsymbol{w_i}$. (also refer to our discussion in Section \ref{illustrative}).


The design of rx-combining vectors is a crucial aspect of this work.
As an initial step towards obtaining proper rx-combining vectors, we consider designing $\boldsymbol{w_i}$'s using linear summation of all analog beamformers that correspond to \textit{beam}$_{\text{\textit{rx}}}$\textit{\#j}'s $\forall j : h_{i,j} {=} 1$.
Let $\Omega_j {=} \cos(\phi_j) {=} \frac{j}{L_r} \: \forall j {\in} \{0,\dots,n_r{-}1\}$, such that $\boldsymbol{e_r}(\Omega_j)$ is the spatial signature of $\textit{beam}_{\textit{rx}}$\textit{\#j}.
Then, $\boldsymbol{w_i}$ can be designed as
\begin{equation}
\boldsymbol{w_i} = \sum_{j=0}^{n_r-1} \mathbbm{1}(h_{i,j}  {=} 1) \boldsymbol{e_r}(\Omega_j)
\end{equation}

\subsection{Sufficient Statistic}
\label{SufficientStatistic}
We will show in this section that each channel syndrome can only be mapped to a single \textit{measurable channel}. A measurable channel in this context refers to $n_r{\times}1$ channels with $L$ non-zero components such that $L \leq e_n$, where $e_n$ is the error correction capability of the underlying code $C$ and $n_r {=} n$ is its CWs length.


Since each measurement combines signals coming from multiple directions, each element in the channel syndrome vector is a linear combination of a subset of the available paths. In other words, each measurement has the possibility that one or more paths are included in it.
This setting is rather challenging. To understand why, consider a channel that has two paths with gains $\alpha_1,\alpha_2 \in \mathbb{C}$. Suppose that
$\alpha_1$ and $\alpha_2$ are of equal magnitudes but are out-of-phase (i.e., phase shift $= 180^\circ$). Hence, if signals coming from both paths are combined in a single measurement, the resultant value is $0$ which is similar to the result we get if no paths exist in the measured directions.
Also each channel measurement can be a result of endless possibilities for the combined path gain values.
So, a natural question to ask is: does this ambiguity cause measurement errors?
The direct answer to this question is: \textbf{No}.
In the sequel we will show that the resulting channel syndrome, i.e., the combination of all channel measurements, is sufficient to correctly estimate the channel.

First, recall our discussion in Footnote \ref{LBC}.
Then, consider all \textit{single-bit error patterns} $\boldsymbol{e^{(i)}}$ 
of a code $C$, with maximum number of correctable errors ${=}e_n$, such that
\[
e_k^{(i)} = \begin{cases}
               1, & k = i\\
               0, & k \neq i
            \end{cases}
\]
where $e_k^{(i)}$ is the $k^{th}$ component of $\boldsymbol{e^{(i)}}$.
Also let $\boldsymbol{s^{(i)}}$ be the corresponding error syndrome of $\boldsymbol{e^{(i)}}$.
Recall that $\boldsymbol{s^{(i)}} {=} \boldsymbol{e^{(i)}}\boldsymbol{H}^T$.
Hence, we can see that $\boldsymbol{s^{(i)}}$ is exactly the $i^{th}$ row of $\boldsymbol{H}^T$, i.e., $i^{th}$ column of $\boldsymbol{H}$.
Now, write any correctable error pattern $\boldsymbol{e}$ as a linear combination of all single-bit error patterns over the finite field $GF(2)$ such that
\begin{equation}
\boldsymbol{e} = \omega_1 \boldsymbol{e}^{(1)} + \omega_2 \boldsymbol{e}^{(2)} + \dots + \omega_n \boldsymbol{e}^{(n)}
\end{equation}
and its corresponding error syndrome is
\begin{equation}
\boldsymbol{s} = \omega_1 \boldsymbol{s}^{(1)} + \omega_2 \boldsymbol{s}^{(2)} + \dots + \omega_n \boldsymbol{s}^{(n)}
\end{equation}
where $\omega_i \in \{0,1\}$ and
$|\omega_i : \omega_i =1 | \leq e_n$.
\begin{lemma}
For an error pattern $\boldsymbol{e}_t$ with number of bit errors identical to $e_n$, its syndrome $\boldsymbol{s}_t$ is a linear combination of $e_n$ linearly independent vectors $\boldsymbol{s}^{(i)}$.
\end{lemma}
\begin{proof}
We are going to prove this lemma by contradiction.
First, assume that $\boldsymbol{s}_t$ is a linear combination of $e_n$ linearly \textbf{dependent} vectors $\boldsymbol{s}^{(i)}$ over $GF(2)$.
Therefore, there exists another error syndrome $\boldsymbol{s}_t^*$ composed of only linear combination of independent vectors $\boldsymbol{s}^{(i)}$ such that $\boldsymbol{s}_t = \boldsymbol{s}_t^*$.
Therefore, there exists another error patter $\boldsymbol{e}_t^*$ with number of errors strictly less than $e_n$ such that its syndrome $\boldsymbol{s}_t^* = \boldsymbol{s}_t$. Since $\boldsymbol{e}_t^*$ has a number of errors less than $e_n$, then it is a correctable error pattern, and since all error syndromes of correctable error patterns are different, then $\boldsymbol{s}_t^*$ should be $\neq \boldsymbol{s}_t$. Hence, we arrive at a contradiction.
\end{proof}
It is also easy to see that if $\boldsymbol{e}_{t_1}$ and $\boldsymbol{e}_{t_2}$ are two different correctable error patterns, then their error syndromes $\boldsymbol{s}_{t_1}$ and $\boldsymbol{s}_{t_2}$ are composed of a linear combination of different sets of single-bit error syndromes $\boldsymbol{s}^{(i)}$.

\begin{lemma}
Any $n-$dimensional linearly independent vectors over $GF(2)$, are also linearly independent over $\mathbb{C}^n$.
\end{lemma}
\begin{proof}
Let $\boldsymbol{v_1}, \dots, \boldsymbol{v_m}$ be a set of $n-$dimensional vectors defined over $GF(2)$. The vectors $\boldsymbol{v_i}$ can be made the columns of an $n {\times} m$ matrix $\boldsymbol{\Psi}$. Since all $\boldsymbol{v_i}$'s are linearly independent over $GF(2)$, then $\boldsymbol{\Psi}$ is an invertible matrix. Therefore, there exists a non-zero (modulo $2$) $m {\times} m$ minor of $\boldsymbol{\Psi}$. Now, suppose the entries in $\boldsymbol{\Psi}$ are interpreted as real numbers. Therefore, $\boldsymbol{\Psi}$, now taken over $\mathbb{R}$, has an $m {\times} m$ sub-matrix whose determinant is non-zero, which proves that it is invertible. Therefore, the vectors $\boldsymbol{v_i}$'s, i.e., columns of $\boldsymbol{\Psi}$, are linearly independent over $\mathbb{R}$ which, using the same argument, can also be shown to be linearly independent over $\mathbb{C}$.
\end{proof}

Suppose that entries of $\boldsymbol{H}$ and $\boldsymbol{e^{(i)}}$ are interpreted as real numbers, then we can write the channel $\boldsymbol{q^a}$ as
\begin{equation}
({\boldsymbol{q}^a})^T = \alpha_1 \boldsymbol{e}^{(1)} + \alpha_2 \boldsymbol{e}^{(2)} + \dots + \alpha_n \boldsymbol{e}^{(n)}
\end{equation}
where $\alpha_i {\in} \mathbb{C}$
and $\sum_{i=1}^n \mathbb{\mathbbm{1}}(\alpha_i {\neq} 0) {\leq} e_n$.
Therefore, each channel syndrome ${(\boldsymbol{y}_s)}^T {=} {(\boldsymbol{q}^a)}^T \boldsymbol{H}^T \Longrightarrow \boldsymbol{y}_s {=} \boldsymbol{H} \boldsymbol{q}^a$ is a linear combination of independent vectors in $\mathbb{C}^{n{-}k}$ (columns of $\boldsymbol{H}$). Therefore, all possible measurable channels yield unique channel syndromes which implies that they are sufficient for the channel estimation problem.

\subsection{Mapping Function $\xi()$}
\label{Mapping}
Now that we have shown that each measurable channel can be mapped to a unique channel syndrome, we need to find this mapping function, i.e., $\xi {:} \boldsymbol{y_s} \rightarrow \widehat{\boldsymbol{q}}^{\boldsymbol{a}}$, where $\widehat{\boldsymbol{q}}^{\boldsymbol{a}}$ denotes the estimated channel.
Again, we resolve to a technique used in hard decision decoding where a look-up table is constructed that maps every error syndrome to a corresponding error pattern. Likewise, we construct a look-up table that indicates which channel corresponds to an obtained channel syndrome.

Since we employ ADCs with finite resolution,
only a finite number of realizable syndromes, $\boldsymbol{y_s}$, exist (and a finite number of corresponding channels).
Therefore, a look-up table method is feasible. We construct the table by, first, generating all possible sparse angular channels.
Then, we find the corresponding channel syndromes using $\boldsymbol{y_s} = \boldsymbol{H} \boldsymbol{q}^a$, where $\boldsymbol{q}^a \in \mathcal{Q}^a$ such that $\mathcal{Q}^a$ is the set of all measurable channels.
Let the actual, arbitrary, noise-corrupted, received channel syndrome be given by $\boldsymbol{y_s^r}$. Therefore, $\boldsymbol{y_s^r}$ might not match exactly one of the channel syndrome vectors in the look-up table. Hence, we instead search for the $\boldsymbol{y_s}$ table entry that has the closest \textit{distance} $\delta$ to $\boldsymbol{y_s^r}$, and pick its corresponding channel as the estimated channel $\widehat{\boldsymbol{q}}^{\boldsymbol{a}}$.
We define the distance between the two complex $m-$dimensional vectors $\boldsymbol{y_s},\boldsymbol{y_s^r}$, as follows:
\begin{equation}
\delta(\boldsymbol{y_s},\boldsymbol{y_s^r}) = \sqrt{\sum_{i=0}^{m-1} \lvert y_{s_i} - y_{s_i}^r \rvert ^2 }.
\end{equation}

By obtaining $\widehat{\boldsymbol{q}}^{\boldsymbol{a}}$, we not only identify the AoA at the Rx, but we also obtain the magnitude and phase information associated with every strong path to the TX.



\subsection{Multiple Transmit and Receive Antennas}
\label{MultiTXRX}
So far, we have considered channels with single transmit antennas and shown how to perform beam discovery at RX. To extend our approach to a general setting, we consider channels with $n_t$ antennas at TX, and $n_r$ antennas at RX.
Thus, instead of the TX just sending signals omnidirectionally, now it can perform highly directional transmission.
Recall that the RX is able to perform channel measurements using multi-armed beams. Similarly, the TX can send signals using multi-armed beams to simultaneously focus on multiple directions using precoding vectors $\boldsymbol{f_j}$.

The design of precoding vectors can also be obtained using an LBC approach. Similar to the method of designing rx-combining vectors $\boldsymbol{w_i}$, we look for an LBC, $C_2$, that has CWs of length $n_2 {=} n_t$ and can correct for $e_n {=} L$ errors. Let the parity check matrix of $C_2$ be $\boldsymbol{H}_2$, using which, we will design the precoding vectors $\boldsymbol{f_j}$.
Let \textit{beam}$_{\text{\textit{tx}}}$\textit{\#i} denote the $i^{th}$ TX beam which points to TX direction \textit{dir}$_{\text{\textit{tx}}}$ \textit{\#i}.
Then, just as before, we envisage $\boldsymbol{H}_2$ as an array whose columns are associated with resolvable TX directions such that: i) its $j^{th}$ column corresponds to \textit{dir}$_{\text{\textit{tx}}}$\textit{\#j}, and ii) its $i^{th}$ row corresponds to the $i^{th}$ measurement.
We note that no actual measurements are performed at TX; we use the word \textit{measurement} to refer to precoding, consistent with the case of RX.
That is, the $i^{th}$ TX measurement is actually the $i^{th}$ precoder $\boldsymbol{f}_i$.
Thereby, we design the $i^{th}$ precoder as a multi-armed TX beam such that, only if $h_{i,j}$, the intersection of the $i^{th}$ row and $j^{th}$ columns of $\boldsymbol{H}_2$, is $ = 1$, do we include sub-beam \textit{beam}$_{\text{\textit{tx}}}$\textit{\#j} in $\boldsymbol{f_i}$.
Each TX measurement provides a component in a TX channel syndrome vector $\boldsymbol{y_s^{TX}}$. The total number of TX measurements (i.e., precoding vectors), denoted by $m_2$, is equal to the number of parity check bits of the code $C_2$. That is, $m_2 {=} n_2{-}k_2$, where $k_2$ is the length of $C_2$'s information sequences.
To obtain AoDs of strong paths at TX, we define the function $\xi_2()$ as the mapping function between all possible TX channel syndromes and their corresponding angular channels denoted by ${\boldsymbol{q}^a}^{\boldsymbol{TX}}$.
Note that, for every \textit{dir$_{\text{\textit{rx}}}$\#i}, there exists a corresponding ${\boldsymbol{q}^a}^{\boldsymbol{TX(i)}}$ which represents the $i^{th}$ row of $\boldsymbol{Q^a}$.
Also, since the maximum number of clusters is $L$, then, the number of non-zero vectors ${\boldsymbol{q}^a}^{\boldsymbol{TX(i)}}$ is $\leq L$.

\IncMargin{1em}
\begin{algorithm}[t]
\SetKwData{Left}{left}\SetKwData{This}{this}\SetKwData{Up}{up}
\SetKwFunction{Union}{Union}\SetKwFunction{FindCompress}{FindCompress}
\SetKwInOut{Input}{input}\SetKwInOut{Output}{output}

 \Input{$ \{ \boldsymbol{w}_i\}_{ \forall i \in \{1,\dots,m_1\}} \:$,
 $\{ \boldsymbol{f}_j \} _{, \forall j \in \{1,\dots,m_2\}} \:$ ,
 $\xi_1()  :  \boldsymbol{y_{s}}  \rightarrow \widehat{\boldsymbol{q}}^a \:$,
 $\xi_2()  :  \boldsymbol{y_{s}^{TX}}  \rightarrow \hat{\boldsymbol{q}}^{a^{\boldsymbol{TX}}}$
 }
 \Output{ $\{\boldsymbol{y_{s_i}}\}_{\forall i \in \{1,\dots,m_2\}}$ }
\Begin{
 $j = 0$\;
 \While{$j < m_2 $}{
 	$i = 0$\;
	\While{$i < m_1 $}{
		$y_{s_{i,j}} = \boldsymbol{w}_i^H \boldsymbol{Q} \boldsymbol{f}_j s + \boldsymbol{w}_i^H \boldsymbol{n}$  \tcp*[r]{channel measurement}
		$i \leftarrow i + 1$
	}
	$\boldsymbol{y_{s_j}} \leftarrow \{y_{s_{i,j}}\}_{\forall i \in \{1,\dots,m_1\}}$ \tcp*[r]{construct channel syndrome $\boldsymbol{y_{s_j}}$}
		\tcc{find corresponding channel $\boldsymbol{q^a}^{(j)} = [{q_1^a}^{(j)}, {q_2^a}^{(j)}, \dots, {q_{n_r}^a}^{(j)}]^T$}
	$\boldsymbol{q^a}^{(j)} \leftarrow \xi_1(\boldsymbol{y_{s_j}})$\;

		\For{$p\leftarrow 1$ \KwTo $n_r$}{

			\tcc{construct TX channel syndromes $\boldsymbol{y_{s}^{TX(p)}}$, where $\boldsymbol{y_{s}^{TX(p)}} = [y_{s_1}^{TX(p)},y_{s_2}^{TX(p)}, \dots, y_{s_{m_2}}^{TX(p)}]^T$}		
			$y_{s_j}^{TX(p)} \leftarrow {q_p^a}^{(j)}$

			}

		$j \leftarrow j + 1$\;
 }
	\For{$p\leftarrow 1$ \KwTo $n_r$}{
 $\boldsymbol{{q^a}^{TX(p)}} \leftarrow \xi_2(\boldsymbol{y_s^{TX(p)}}) $
 }
 $\widehat{\boldsymbol{Q}}^{\boldsymbol{a}} =
 \begin{pmatrix}
  \boldsymbol{{q^a}^{TX(1)}} & \boldsymbol{{q^a}^{TX(2)}} & \hdots & \boldsymbol{{q^a}^{TX(n_r)}}
 \end{pmatrix}^T
 $
}
 \caption{\small Beam discovery of multiple TX/RX antennas.}
 \label{Alg:MultiTXRX}
\end{algorithm}
\DecMargin{1em}

To see the whole picture, assume that a code $C_1$, with CWs of length $n_1 {=} n_r$, is an LBC code associated with beam discovery at RX side. Let the number of RX measurements, i.e., the number of rx-combining vectors, be $m_1$ such that $m_1  {=} n_1 {-} k_1$, where $k_1$ is the length of information sequences of $C_1$. Also let $\xi_1()$ be the mapping function between RX channel syndromes and its corresponding angular channel.
Under this setting, the beam discovery problem is performed as follows: 
i) The TX starts starts sending its training sequence using the precoder $\boldsymbol{f_j}, \forall j \in \{0,\dots,m_2-1\}$.
ii) The RX performs $m_1$ channel measurements while $\boldsymbol{f_i}$ is being used at TX and obtains a channel syndrome $\boldsymbol{y_{s_j}}$.
iii) Based on $\boldsymbol{y_{s_j}}$, the RX obtains a corresponding channel, $\boldsymbol{{q^a}^{(j)}}$ with path components $\{{q_p^a}^{(j)}\}_{\forall p \in \{1,\dots,n_r\}}$.
Notice that $\boldsymbol{{q^a}^{(j)}}$'s do not necessarily represent individual path gains, but rather, combinations of paths accumulating at a single \textit{dir$_{\text{\textit{rx}}}$\textit{\#}}. Therefore, there exists a resemblance to channel syndromes which we exploit.
iv) We construct a set of $n_r$ TX channel syndromes, $\boldsymbol{y_s^{TX(p)}}$ where their $j^{th}$ component $y_{s_j}^{TX(p)} {=} {q_p^a}^{(j)}$, i.e.,
$[\boldsymbol{y_s^{TX(1)}}, \boldsymbol{y_s^{TX(2)}}, \dots, \boldsymbol{y_s^{TX(n_r)}}] =
[\boldsymbol{{q^a}^{(1)}}, \boldsymbol{{q^a}^{(2)}} , \dots, \boldsymbol{{q^a}^{(m_2)}}]^T$.
v) Finally, we find the angular TX channel for every \textit{dir}$_\text{\textit{rx}}$ \textit{\#p}, i.e., $p^{th}$ row of $\boldsymbol{Q^a}$, using the mapping function $\boldsymbol{{q^a}^{TX(p)}} {=} \xi_2(\boldsymbol{y_s^{TX(p)}})$.
Notice that, since no more than $L \ll n_r$ clusters exist, and since $\boldsymbol{0}$ channels correspond to $\boldsymbol{0}$ channel syndromes, we only need to apply $\xi_2()$ a maximum of $L$ times -unless measurement error occurs.
This whole process is highlighted in Algorithm \ref{Alg:MultiTXRX}.

\begin{remark}{}
The estimated channel $\widehat{\boldsymbol{Q}}^{\boldsymbol{a}}$ may contain more than $L$ non-zero components. The reason is that the receiver obtains a channel $\boldsymbol{{q^a}^{(j)}}$ for every precoder $\boldsymbol{f_j}$ which may contain erroneous component estimates. Incorrect estimates occur as a result of measurement errors which happen due to i) channel noise, ii) quantization error.
Now every $\boldsymbol{{q^a}^{(j)}}$ may contain a maximum of $L$ non-zero components, however, some of which may be due to measurement errors. Afterwards, potentially noise-corrupted $\{\boldsymbol{{q^a}^{(j)}}\}_{\forall j\in\{1,\dots,n_r\}}$ are used to obtain TX channel syndromes as shown in Algorithm \ref{Alg:MultiTXRX}. That is, for every \textit{dir}$_{\text{\textit{rx}}}$\textit{\#i} we obtain a TX channel syndrome to identify the corresponding \textit{dir}$_{\text{\textit{tx}}}$\textit{\#j}'s that have strong components. Thus, we may obtain a maximum of $L$ non-zero components per \textit{dir}$_{\text{\textit{rx}}}$\textit{\#i}.
\end{remark}

\section{Performance Evaluation}
\label{Eval}
\subsection{Performance Metrics}
\begin{figure*}[t]
\centering
\begin{subfigure}{.33\textwidth}
  \centering
  \captionsetup{justification=centering}
\includegraphics[width=0.95\linewidth]{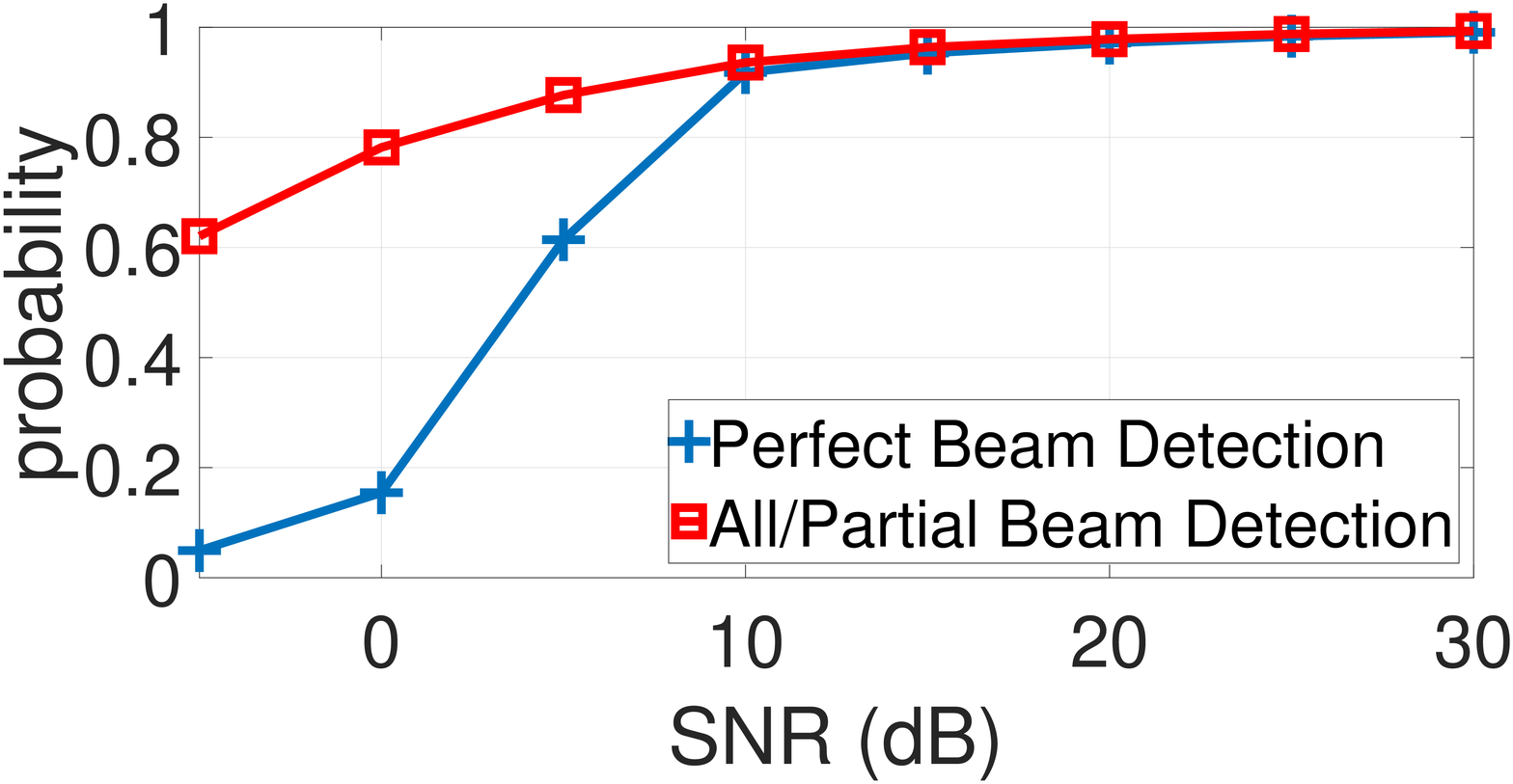}
\caption{\small Beam detection probability \\($15 {\times} 15$ channel with $L{=}1$).}
\label{fig:beamDetProb_15x15_9L_1P}
\end{subfigure}%
\begin{subfigure}{.33\textwidth}
  \centering
  \captionsetup{justification=centering}
\includegraphics[width=0.95\linewidth]{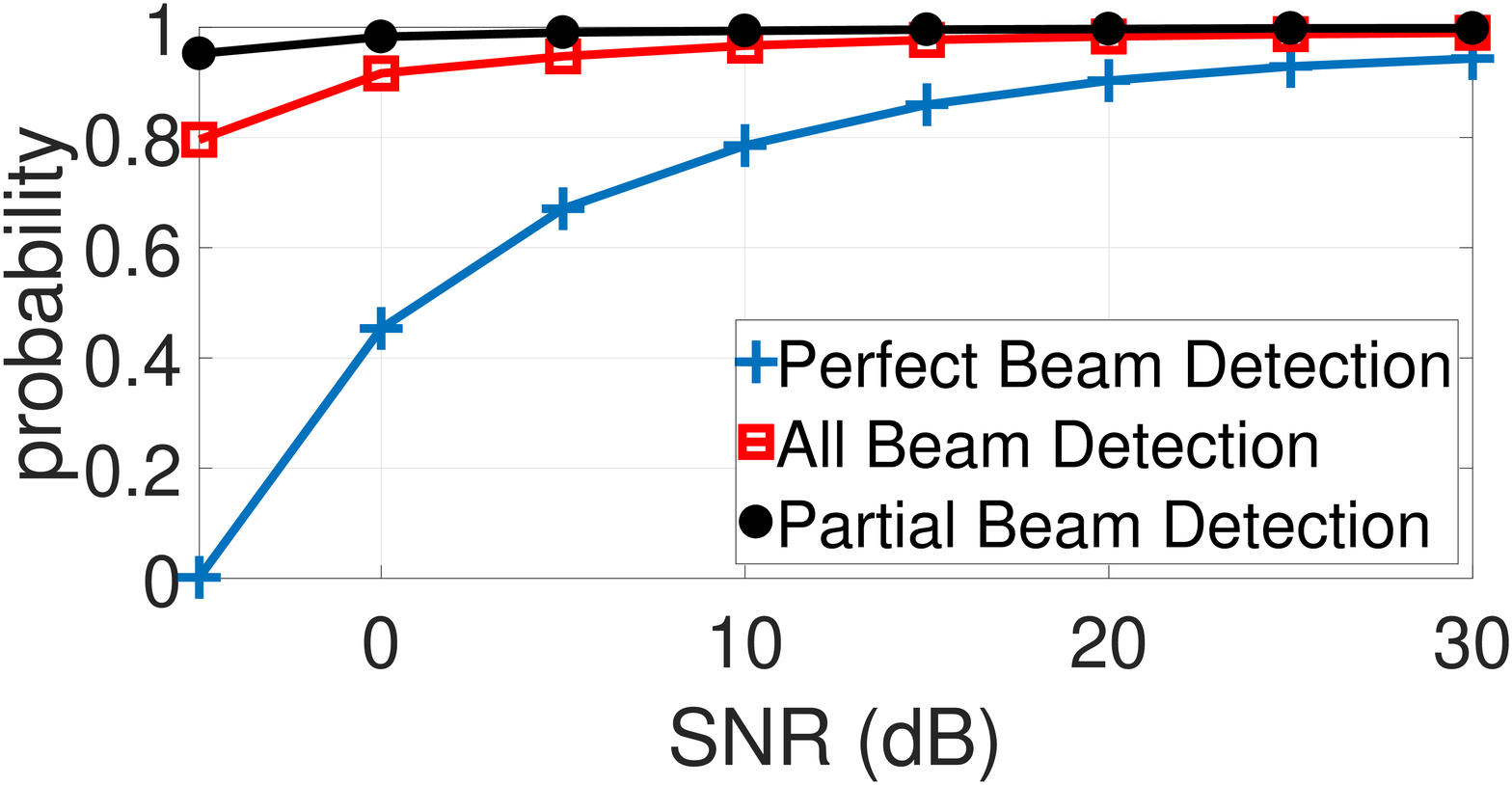}
\caption{\small Beam detection probability \\($8 {\times} 8$ channel with $L{=}2$).}
\label{fig:beamDetProb_8x8_9L_2P} 
\end{subfigure}%
\begin{subfigure}{.33\textwidth}
  \centering
\includegraphics[width=0.95\linewidth]{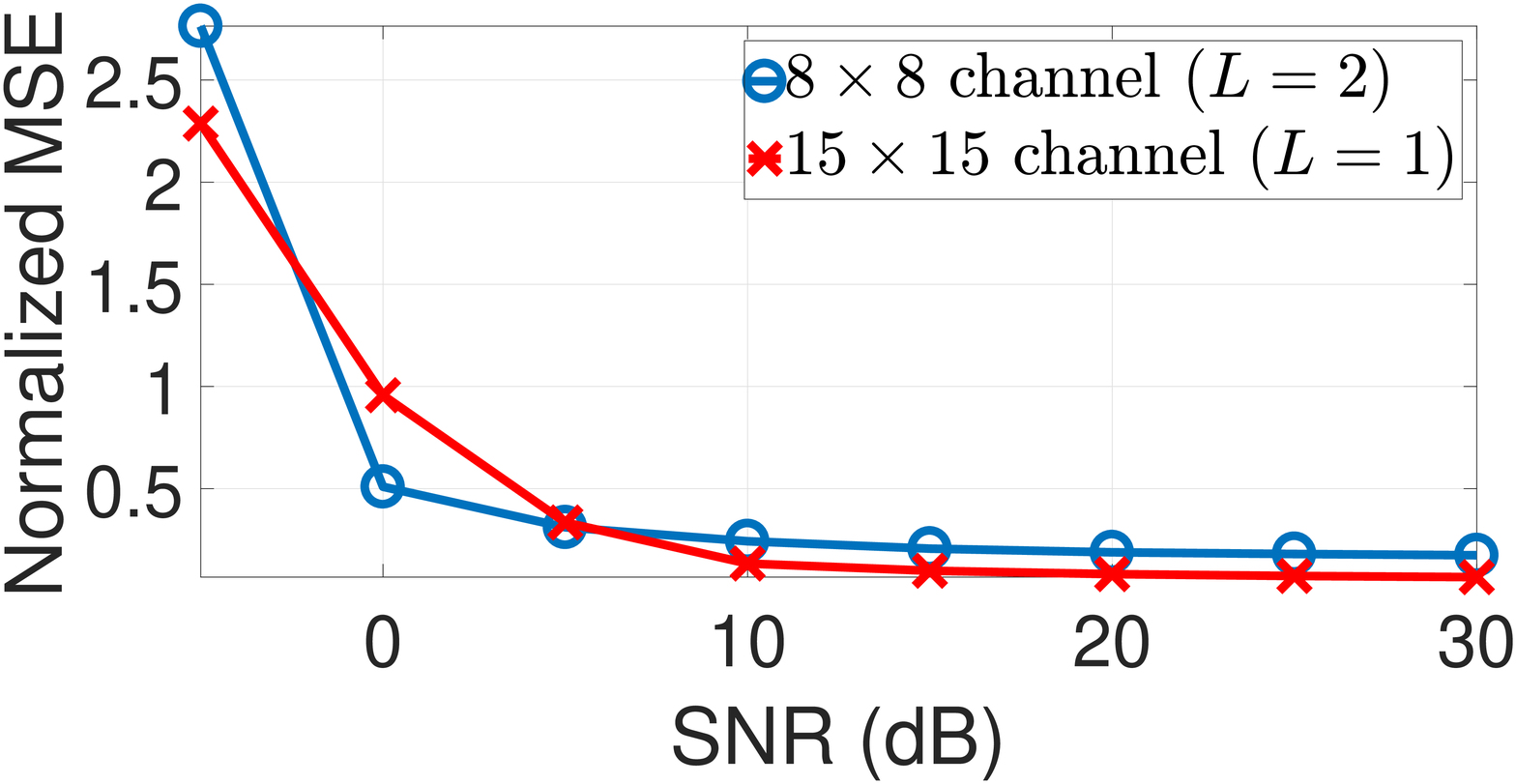}
\caption{Normalized Mean Squared Error (MSE).\\ \textcolor{white}{.}}
\label{fig:MSE}
\end{subfigure}%
\caption{Performance Evaluation.}
\end{figure*}
Our performance metrics focus on three basic criteria; accuracy of beam detection, number of measurements, and accuracy of path gain value estimates. To that end, we use the following performance parameters:\\
i)   \textbf{Number of measurements.} An important performance objective is to achieve a similar performance as the exhaustive search, yet with a much smaller number of measurements.\\
ii)  \textbf{Perfect beam detection probability:} probability that \textbf{only} correct AoAs/AoDs of all available detectable paths are identified (i.e, no incorrect paths are detected along with true ones).\\
iii) \textbf{All-beam detection probability:} probability of correctly identifying all true detectable paths (with possibility of other incorrect paths).\\
iv)  \textbf{Partial beam detection probability:} probability that at least one correct path is detected (with possibly other incorrect paths).\\
v)   \textbf{Number of incorrect beams:} the number of incorrect paths identified between TX and RX.\\
vi)  \textbf{Normalized mean squared error (MSE):} $\frac{\left\Vert \boldsymbol{Q^a} - \widehat{\boldsymbol{Q}}^{\boldsymbol{a}}\right\Vert^2_F}{\left\Vert \boldsymbol{Q^a}\right\Vert^2_F}$.\\
We assess these performance parameters against different \textit{SNR} levels which capture the effect of noise power on performance.

\subsection{Simulation Parameters}
Unless otherwise specified, we use the following simulation parameters. We consider a mm-wave channel with $L$ paths between TX and RX\footnote{Note that we always generate $L$ channel paths, however, according to the preset \textit{SNR}, some path gains are too weak to be detected at RX (recall that in Eq. \ref{eqn:SNRl}, \textit{SNR} is defined per path).
Thus, only paths with enough signal power to cross the threshold of the $1^{st}$ level (${\pm} 1$) of the utilized ADCs are considered as the ground truth to which we evaluate our performance parameters.}.
Path gains $\alpha_l^b$ are generated such that $\textit{SNR}_l$ lies in a range between $\textit{SNR}_{min}$ and ($\textit{SNR}_{min} {+} 20$) dB for paths that contain strong reflectors.
The noise power is $N_0 {=} {-}95$dBm. 	The signal power $P$ is adjusted such that for the weakest reflector, an $\textit{SNR}_{min}$ value is achieved. For example, at \textit{SNR} ${=} 10$dB, the required transmitted power for a $15 {\times} 15$ channel is $P {\approx} 43$dBm and for an $8 {\times} 8$ channel is ${\approx} 46$dBm.
Further, we use ADCs with $9$ quantization levels, i.e., $\{-4,\dots,4\}$, where the maximum output is adjusted to the maximum expected received signal power according to the preset \textit{SNR}. Note that the maximum received power depends on the available number of paths. Also negative quantization levels, e.g. $-4$, represent signals with the same power level as $4$ but with phase shift $= 180^\circ$ where the phase reference can be set arbitrarily at RX. For every simulation scenario, we average across $10^5$ runs.

\subsection{SinglePath Channels}
%
%
\begin{figure}[t]
\centering
\includegraphics[width=1\linewidth]{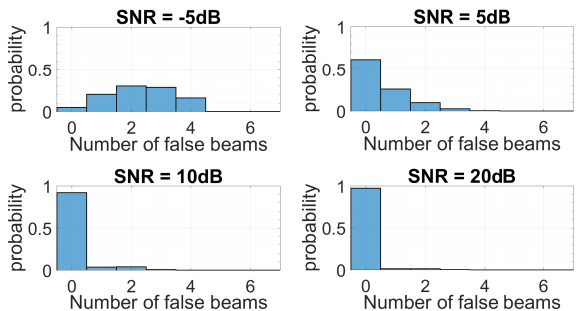}
\caption{Number of incorrect beams ($15 {\times} 15$ channel with $L{=}1$)}
\label{fig:Hist_2_15x15_9L_1P}
\end{figure}

Consider a $15 {\times} 15$ mm-wave channel with $L {=} 1$ path between TX and RX. Hence, the parity check matrix of $(15,11,3)$ Hamming code can be used for the design of both precoders, $\boldsymbol{f_j}$, and rx-combiners, $\boldsymbol{w_i}$, i.e., $\boldsymbol{H_1}$ and $\boldsymbol{H_2}$, respectively.
Hence, we need a number of TX measurements $m_1$, which is identical to the number of RX measurements $m_2 {=} 15{-}11 {=} 4$. Hence, the total number of measurements is $m {=} 16$. On the other hand, the exhaustive search method needs $255$ measurements to check every possible TX and RX beam combination. Thus, our approach results in ${\sim} 92.8 \%$ reduction in the number of measurements needed for beam discovery.
We investigate the probability of perfect, all and partial beam detection at different levels of \textit{SNR}.
Fig. \ref{fig:beamDetProb_15x15_9L_1P} depicts these probabilities. Note that since there only exists $1$ path, both all and partial beam discoveries are the same.

Further, we investigate the performance in terms of the probability of obtaining a specific number of incorrect paths/beams. This is depicted as a histogram plot in Fig. \ref{fig:Hist_2_15x15_9L_1P}.
We notice that the number of incorrect beams decreases with increasing \textit{SNR}.
For instance, at SNR ${=}-5$dB, the probability of obtaining $0$ incorrect paths is ${\sim} 0.05$. This probability starts to increase
until it peeks at ${\sim} 0.3$ for $2$ incorrect paths. Afterwards it starts decreasing again and reaches $\sim 0.16$ at $4$ incorrect paths.
At higher levels of \textit{SNR}, the probability of obtaining less number of incorrect paths increases dramatically, e.g., at \textit{SNR} ${=}20$dB, there is a probability of $0.972$ of obtaining no incorrect paths where the remaining $0.028$ is divided between $1,2$ and $3$ incorrect paths in a decreasing fashion.

Finally, Fig. \ref{fig:MSE} depicts the normalized mean squared error of the channel estimate $\widehat{\boldsymbol{Q}}^{\boldsymbol{a}}$. The very high MSE value at ${-}5$dB indicates very poor channel estimates which indicates that $\widehat{\boldsymbol{Q}}^{\boldsymbol{a}}$ has large components at truly zero components in $\boldsymbol{Q^a}$. Nevertheless, MSE drops fast at higher values of \textit{SNR}. We can see that starting at \textit{SNR} $ {=} 10$dB, the performance becomes highly reliable.

\subsection{MultiPath  Channels}
%
\begin{figure}[t]
\centering
\includegraphics[width=1\linewidth]{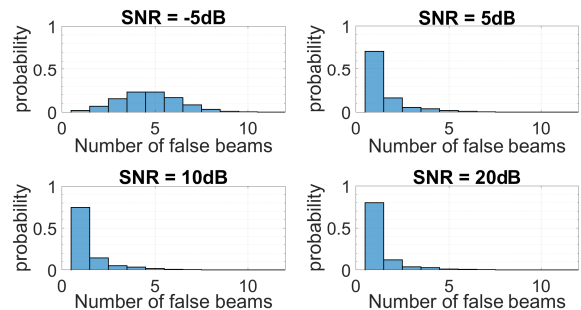}
\caption{Number of incorrect beams ($8 {\times} 8$ channel with $L{=}2$)}
\label{fig:Hist_2_8x8_9L_2P} 
\end{figure}
Consider an $8 {\times} 8$ channel with $L {=} 2$ paths. Thus, an $(8,2,5)$ code is a proper choice for both $\boldsymbol{H_1}$ and $\boldsymbol{H_2}$. A total number, $36$, of measurements is needed for beam discovery. As opposed to $64$ measurements needed for exhaustive search, we achieve ${\approx} 43.7 \%$ reduction in the number of measurements under this scenario.
Fig. \ref{fig:beamDetProb_8x8_9L_2P} shows the probability of perfect, all, and partial beam detection.
We notice a large gap between these probability measures at low \textit{SNR}. However, this gap becomes quickly narrower as \textit{SNR} increases.
For the number of incorrect paths, we obtain similar performance trends to the single-path channel case. However, the maximum number of incorrect paths experienced is slightly larger.
Histograms that depicts the probabilities of the number of incorrect beams are shown in Fig. \ref{fig:Hist_2_8x8_9L_2P}.
Again, we observe that the reliability of beam discovery increases significantly with increasing \textit{SNR} which is depicted in terms of MSE in Fig. \ref{fig:MSE}.

\section{Conclusion}
\label{DiscConc}
This work provides a solution for the mm-wave channel estimation problem by exploiting its sparse nature in the angular domain. The proposed solution is a beam discovery technique that is similar to error discovery in channel coding. We show that our proposed technique can significantly reduce the number of measurements required for reliable channel estimation.
Our solution takes into account the size of the channel and its sparsity level when determining the number of measurements.
As future work, we will investigate more efficient ways of executing the mapping function $\xi()$, possible improvements for designing precoders and rx-combiners, and mitigating the effect of noise-corrupted channel measurements.



\end{document}